\documentclass[11pt]{article}
\usepackage{graphicx}
\usepackage{amsfonts,amsmath,amssymb,amscd,enumitem,color}
\usepackage{mathtools}
\usepackage{accents}
\usepackage{amsthm}
\usepackage[hidelinks]{hyperref}
\usepackage[capitalize]{cleveref}
\usepackage[english]{babel}
\usepackage[utf8]{inputenc}
 
\newtheorem{theorem}{Theorem}

\newtheorem{lemma}[theorem]{Lemma}

\newtheorem{corollary}[theorem]{Corollary}

\numberwithin{subcase}{case}

\numberwithin{subsubcase}{subcase}

\usepackage{fullpage}
\usepackage{thmtools}
\usepackage{thm-restate}

\author{Brett Leroux, Luis Rademacher}

\def\keywords#1{\par\addvspace\medskipamount{\rightskip=0pt plus1cm
\def\and{\ifhmode\unskip\nobreak\fi\ $\cdot$
}\noindent{Keywords:}\enspace\ignorespaces#1\par}}

\title{Euclidean distance compression via deep random features}

\usepackage{dirtytalk}

\def\final{1}  % set this to 1 to get a comment-free version

\ifnum\final=0  %namely if we allow comments in the output
\newcommand{\lnote}[1]{[{\small Luis: \bf #1}]}
\newcommand{\bnote}[1]{[{\small Brett: \bf #1}]}
\newcommand{\anonnote}[1]{[{\small anon: \bf #1}]}
\newcommand{\sidecomment}[1]{\marginpar{\tiny #1}}
\newcommand{\details}[1]{{\color{blue}\ [[#1]] }}
\else % in this case [final=1] we don't want any comments to show
\newcommand{\lnote}[1]{}
\newcommand{\bnote}[1]{}
\newcommand{\anonnote}[1]{}
\newcommand{\sidecomment}[1]{}
\newcommand{\details}[1]{}
\fi  %ok, we're done with defining comment macros
\newcommand{\Rl}{\operatorname{\mathbb{R}}}

\newcommand{\Tor}{\operatorname{\Tor}}

\newcommand{\sign}{\operatorname{sign}}

\newcommand{\suchthat}{\mathrel{:}}

\newcommand{\e}{\operatorname{\mathbb{E}}}

\newcommand{\prob}{\mathbb{P}}

\newcommand{\norms}[1]{{\lVert#1\rVert}^2}

\makeatletter
\newcommand\fake@math{}% just for safety
\def\fake@math#1\){[math]}
\makeatother
\usepackage{color}

\newcommand{\abs}[1]{\lvert#1\rvert}

\newcommand{\NN}{\mathbb{N}}

\begin{document}
\maketitle

\begin{abstract}
%A consequence of the Johnson-Lindenstrauss lemma is that if a set of data points is mapped by a random projection matrix to some lower dimensional space, then pairwises distances between data points are approximately preserved. Moreover, if one wants to guarantee with high probability that pairwise distances are preserved up to a multiplicative $(1\pm \epsilon)$ error, then the Johnson-Lindenstrauss lemma shows that it suffices to take the dimension of the projected points to be $\Theta(\epsilon^{-2}\log n)$ where $n$ is the number of points. If $\epsilon$ is not too small, random projection compresses the data set while maintaining the ability to estimate pairwise distances between data points.

%This work introduces a new type of random mappings which accomplish the same\lnote{discuss. not good to say same, but better} objective. We construct random nonlinear maps $\varphi_\ell:\Rl^d \to \{-1,1\}^N$ such that, with high probability, storing $\varphi(S)$ allows one to report pairwise squared distances between points in $S$ up to some multiplicative $(1\pm \epsilon)$ error. As in the case of the Johnson-Lindenstrauss lemma, the required dimension $N$ of the image of the map depends on $\epsilon$ and the number of points. The main advantage of our maps $\varphi_\ell$ over random linear maps is that ours map points\lnote{point} sets directly into the discrete cube $\{-1,1\}^N$ and so there is no additional step needed to convert the sketch to bits. In some cases, our maps $\varphi_\ell$ produce sketches which require fewer bits of storage space.

Motivated by the problem of compressing point sets into as few bits as possible while maintaining information about approximate distances between points, we construct random nonlinear maps $\varphi_\ell$ that compress point sets in the following way. For a point set $S$, the map $\varphi_\ell:\Rl^d \to N^{-1/2}\{-1,1\}^N$ has the property that storing  $\varphi_\ell(S)$ (a \emph{sketch} of $S$) allows one to report pairwise squared distances between points in $S$ up to some  multiplicative  $(1\pm \epsilon)$ error with high probability as long as the minimum distance is not too small compared to $\epsilon$.  The maps $\varphi_\ell$ are the $\ell$-fold composition of a certain type of random feature mapping. Moreover, we determine how large $N$ needs to be as a function of $\epsilon$ and other parameters of the point set. 

Compared to existing techniques, our maps offer several advantages. The standard method for compressing point sets by random mappings relies on the Johnson-Lindenstrauss lemma which implies that if a set of $n$ points is mapped by a Gaussian random matrix to $\Rl^k$ with $k =\Theta(\epsilon^{-2}\log n)$, then pairwise distances between points are preserved up to a multiplicative $(1\pm \epsilon)$ error with high probability. The main advantage of our maps $\varphi_\ell$ over random linear maps is that ours map point sets directly into the discrete cube $N^{-1/2}\{-1,1\}^N$ and so there is no additional step needed to convert the sketch to bits. For some range of parameters, our maps $\varphi_\ell$ produce sketches which require fewer bits of storage space.
\end{abstract}

\section{Introduction}
Random projection is a commonly used method to lower the dimension of a set of points while maintaining important properties of the data \cite{Santosh}. The random projection method involves mapping a high-dimensional set of points in $\Rl^d$ to a lower dimensional subspace by some random projection matrix in such a way that the pairwise distances and inner products between the projected points are approximately equal to the pairwise distances and inner products between the original points. The random projection method has many applications to data analysis and prominent algorithms such as nearest neighbor search \cite{DBLP:conf/stoc/IndykM98}. 

The theoretical foundation of random projection is the Johnson-Lindenstrauss lemma which states that a random orthogonal projection to a lower dimensional subspace has the property of preserving pairwise distances and inner products \cite{JL}. Later it was observed \cite{MR1943859, DBLP:conf/stoc/IndykM98} that one can alternatively take the projection matrix to be a matrix with i.i.d.\ Gaussian $N(0,1)$ entries. 

\begin{lemma}[\cite{MR1943859, DBLP:conf/stoc/IndykM98,Santosh}]\label{lem:JL}
    Let each entry of an $d \times k$ matrix $R$ be chosen independently from $N(0,1)$. Let $v = \frac{1}{\sqrt{k}}R^Tu$ for $u \in \Rl^d$. Then for any $\epsilon>0$, $\prob\bigl(\bigl\lvert\|v\|^2-\|u\|^2 \bigr\rvert\ge \epsilon\|u\|^2\bigr)<2e^{-(\epsilon^2-\epsilon^3)k/4}$.
\end{lemma}

A corollary of the above lemma is that if an arbitrary set of $n$ points in $\Rl^d$ is mapped by the random projection matrix $R$ to $\Rl^k$ where $k = \Theta(\epsilon^{-2}\log n)$, then the squared distances between pairs of points are distorted by a factor of at most $(1\pm\epsilon)$ with high probability. The projected points are thus a lower dimensional representation of the original point set and this lowering of the dimension offers two main advantages. The first is that algorithms that were originally intended to be performed on the original point set can now instead be performed on the lower dimensional points, in many cases offering significant run-time speed ups while at the same time guaranteeing that the output of the algorithm on the low dimensional points is a good approximation to the output for the original points. This approximation guarantee is usually the result of the fact that the random projection preserves pairwise squared distances.

The second main advantage of the lower dimensional representation is a reduction in the cost of data storage. Let $S$ be a set of $n$ points in the unit ball in $\Rl^d$. 
The Johnson-Lindenstrauss lemma shows that random projection produces a data structure (called a \emph{sketch}) that can be used to report the squared distances between pairs of points in $S$ up to a multiplicative $(1\pm\epsilon)$ error. 
The size of the sketch of course depends on both $|S|$ and $\epsilon$. 
From this viewpoint, it is natural to ask what is the minimum number of bits of such a sketch? The Johnson-Lindenstrauss lemma gives an upper bound on the size of such a sketch as follows: 
Any set of $n$ points $S$ can be projected to $\Rl^k$ where $k = \Theta(\epsilon^{-2}\log n)$ while distorting pairwise squared distances by a factor of at most $(1\pm\epsilon)$. 
The projected data points are real-valued, and thus the number of bits required to store the projected data is not determined solely by the dimension of the random projection. 
Each coordinate of the projected data points needs to be encoded into bits in such a way that guarantees squared distances are preserved. 
One way to convert to bits is to use an epsilon-net for the unit ball in $\Rl^k$: 
In order to preserve squared distances up to a multiplicative $(1\pm\epsilon)$ error, it suffices to preserve squared distances up to an additive $m^2\epsilon$ error where $m$ is the minimum distance between pairs of points in $S$. 
By identifying each projected point with the closest point in an $m^2  \epsilon$-net, we can produce a sketch with $\Theta \bigl(n\epsilon^{-2}\log n\log(1/m^2\epsilon) \bigr)$ bits.\footnote{Since the points are in the unit ball, for any $\delta>0$, approximating distances to within additive error $\delta/6$ gives approximation of squared distances to within additive error $\delta$. So the size of the epsilon net is, up to a constant, $(1/m^2  \epsilon)^{\epsilon^{-2}\log n}$. This gives $\Theta\bigl( \log(1/m^2  \epsilon)\epsilon^{-2}\log n\bigr)$ bits per point}

While the Johnson-Lindenstrauss lemma shows that efficient sketches can be obtained by mapping the points to a lower dimensional space with a random linear mapping (the projection), it is natural to ask if there are other types of random maps (in particular, possibly nonlinear maps) which are able to produce sketches with a smaller number of bits. 
Our main result shows that this is possible in certain cases. We state our main result first for sets of points contained in the unit sphere $S^{d-1}$ and at the end of this section we include the extension to subsets of the unit ball.
%\newtheorem*{theorem:sphere}{\cref{thm:sphere}}
%\begin{theorem:sphere}

\begin{restatable}{theorem}{sphere}\label{thm:sphere}
  Let $S \subset S^{d-1}$ with $|S|=n\ge 2$. Let $m = \min_{x,y \in S,x\neq y}\| x- y\|$ and $\ell = \lceil\log_2\log_2 \frac{4}{m} \rceil \ge 1$. 
  Let $\epsilon>0$ and assume that $\epsilon < \min_{x,y \in S,x\neq y}1-|\langle x,y \rangle|$. 
  Then the random map $\varphi_\ell:S^{d-1 } \to \frac{1}{\sqrt{N}} \{-1,1\}^N$ with $N = \Theta\bigl(\frac{\log n}{\epsilon^2} (\log \frac{1}{m} )^{2\log_2(\pi/\sqrt{2})} \bigr)$ (defined in the proof of \cref{thm:conditional2} and independent of $S$ except through parameters $n, d, m$ and $\min_{x,y \in S,x\neq y}1-|\langle x,y \rangle|$) satisfies the following with probability at least $(1-\frac{2}{n})^\ell$: $\varphi_\ell( S)$ is a sketch of $S$ that allows one to recover all squared distances between pairs of points in $S$ up to a multiplicative $(1\pm \epsilon)$ error. The number of bits of the sketch is $\Theta\bigl(\frac{n \log n}{\epsilon^2} (\log \frac{1}{m} )^{2\log_2(\pi/\sqrt{2})} \bigr)$. 
%\end{theorem:sphere}
\end{restatable}
%\lnote{it'd be nice to say something about Achlioptas, but not clear what to say. Maybe: As discussed in [indyk wagner], an alternative binary compression is possible if the points have bounded integer coordinates and one uses the ``binary'' variant of Johnson-Lindenstrauss due to [Achlioptas]. This approach is somewhat incomparable to our setting because of the integrality assumption.}\bnote{added to 3rd paragraph of 1.3}
%\lnote{on why assumption $\epsilon<1-|\langle \hat x,\hat y \rangle|$ excludes inner products close to $-1$ (i.e. assumption is not $\epsilon<1-\langle \hat x,\hat y \rangle$): add comment to discussion, something like `it is a limitation of our argument, no known lower bound that shows this is needed, open problem to improve'}\bnote{added to 3rd paragraph below}
The proof of \cref{thm:sphere} is in \cref{sec:recovery}. 
We explain how the map $\varphi_\ell$ is constructed in the next subsection and the role of the parameter $\ell$ is discussed in \cref{sec:intuition}. The main advantage of the map $\varphi_\ell$ is that it maps the point set $S$ directly into the discrete cube and thus there is no need to convert the sketch to bits after performing the random mapping. Furthermore, the map $\varphi_\ell$ produces sketches with asymptotically fewer bits than those obtained using the Johnson-Lindenstrauss lemma as long as  $(\log(\frac{1}{m}))^{2\log_2(\pi/\sqrt{2})}=o \bigl( \log(\frac{1}{m^2\epsilon})\bigr)$. This is equivalent to the condition that $(\log(\frac{1}{m}))^{2\log_2(\pi/\sqrt{2})}=o \bigl( \log(\frac{1}{\epsilon}) \bigr)$.

Furthermore, the sketch $\varphi_\ell(S)$ has the desired properties \say{with high probability}. 
The probability that the sketch succeeds is $(1-2/n)^\ell$ and we claim that this quantity approaches one for all useful choices of the parameters: 
First of all, we recall that in all applications of \cref{thm:sphere}, we should take $n>d$. 
We also take $\epsilon = \Omega(d^{-1/2})$ (ignoring logarithmic factors) because otherwise the target dimension $N$ is larger than the original dimension $d$ and then a sketch based on an $\epsilon$-net (as above, but without random projection) for $S$ would be better, by having bitlength $nd$ (up to logarithmic factors).
%\lnote{discuss: but why does it matter, goal is small bitlength, not small dimension}\bnote{change}
The assumption that $\epsilon<1-|\langle x,y\rangle | $ for all $x,y$, $x\neq y$ means that $\epsilon< m^2$ and so $1/m<\epsilon^{-1/2}= O(d^{1/4})$. 
This means that $\ell = O(\log_2 \log_2 d)$ and so $(1-2/n)^\ell> (1-2/d)^\ell$ approaches one. 

We remark that it might be possible to replace the assumption in \cref{thm:sphere} that $\epsilon<1-|\langle x,y \rangle|$ for all $x,y\in S$, $x\neq y$ by the weaker assumption that $\epsilon<1-\langle x,y \rangle=\frac{1}{2}\|x-y\|^2$. The reason that an assumption of this sort is necessary is that, because we are trying to produce sketches which allow recovery of squared distances, pairs of points with very small distance are difficult to deal with. As a result, we need to assume for technical reasons that the accuracy parameter $\epsilon$ is smaller then (half) the minimum squared distance, i.e. $\epsilon <1-\langle x,y \rangle$. There is an inherent symmetry in the maps $\varphi_\ell$ that we use which makes it convenient to use the stronger assumption that $\epsilon <1-|\langle x,y \rangle|$.

We are able to extend our main result to deal with not only sets of points contained in the unit sphere, but also any set of points in the unit ball. For a point $x$ in the unit ball  $B^d$ we use $\hat x$ to denote $x/\|x\| $.
%\newtheorem*{theorem:main}{\cref{thm:main}}
%\begin{theorem:main}
\begin{restatable}{theorem}{main}\label{thm:main}
    Let $S\subset B^d$ with $|S|=n\ge 2$ and set $\rho = \min_{x \in S}\|x\|^2$. Let $m= \min_{x,y \in S,x\neq y}\|\hat x-\hat y\|$ and $\ell = \lceil\log_2\log_2 \frac{4}{m} \rceil \ge 1$.
    Let $\epsilon>0$ and assume that  $\epsilon < \min_{\hat x,\hat y \in S,x\neq y}1-|\langle \hat x,\hat y \rangle|$. 
   Then the random map $\varphi_\ell:S^{d-1 } \to \frac{1}{\sqrt{N}}\{-1,1\}^N$ with $N = \Theta\bigl(\frac{\log n}{\epsilon^2} (\log\frac{1}{m})^{2\log_2(\pi/\sqrt{2})} \bigr)$ (defined in the proof of \cref{thm:conditional2} and independent of $S$ except through parameters $n, d, m$ and $\min_{x,y \in S,x\neq y}1-|\langle \hat x,\hat y \rangle|$) satisfies the following with probability at least $(1-\frac{2}{n})^\ell$:
    $\varphi_\ell(\hat S)$ and the norm of each point in $S$ up to an additive $\pm \rho m^2\epsilon/48 $ error is a sketch of $S$ that allows one to recover all squared distances between pairs of points in $S$ up to a multiplicative $(1\pm \epsilon)$ error.
    Moreover the number of bits of the sketch is $\Theta\Bigl(\frac{n \log n}{\epsilon^2} \bigl(\log \frac{1}{m}\bigr)^{2\log_2(\pi/\sqrt{2})} +  n\log \frac{1}{\rho m^2\epsilon} \Bigr)$.
%\end{theorem:main}
\end{restatable}

The proof of \cref{thm:main} is in \cref{sec:recovery}. The number of bits of the sketches in the above theorem depend on the parameter $m$ which is the minimum distance between pairs of points after all points have been normalized to have unit norm. Thus, it is more complicated to compare the number of bits of our sketches to the sketches obtained using the Johnson-Lindenstrauss lemma because the Johnson-Lindenstrauss sketches do not rely on any normalization step. However, as in the case of point sets on $S^{d-1}$, there is a large family of point sets in $B^d$ for which our sketching technique produces sketches with a fewer number of bits.

\subsection{The maps $\varphi_\ell$ and the recovery of $\|x-y\|^2$ by $\varphi_\ell(x)$ and $\varphi_\ell(y)$}
In this section we summarize the construction of the maps $\varphi_\ell$ which are used in \cref{thm:sphere} and formally analyzed in \cref{thm:conditional2} and how they can be used to recover squared distances between points.
Let $f(t) = \frac{2}{\pi}\arcsin(t)$ and $g(t) = \sin(\frac{\pi t}{2})$. 
So $f:[-1,1] \to [-1,1]$ is the inverse of $g:[-1,1]\to [-1,1]$. 
For $\ell\in \mathbb{N}^+$, let $f_\ell$ be the function $f$ composed with itself $\ell$ times and similar for $g_\ell$. 
Notice that for any $\ell \in \mathbb{N}^+$, $f_\ell:[-1,1]\to [-1,1]$ is the inverse of $g_\ell:[-1,1] \to [-1,1]$. 

For simplicity in the rest of the introduction we assume all points are normalized to be on the unit sphere $S^{d-1}$. We define the sign function as 
\begin{equation*}
\sign(t)=
    \begin{cases}
        1 & \text{if } t \ge 0\\
        -1 & \text{if } t < 0.
    \end{cases}
\end{equation*}
The maps $\varphi_\ell$ will be defined as the composition of $\ell$ maps of the following form. Set $D \in \mathbb{N}^+$. Let $Z_i$, $1\le i\le D$ be i.i.d.\ standard Gaussian random vectors in $\Rl^d$. Let $\varphi^D: \Rl^d\to \Rl^D$ be defined by 
\[
\varphi^D(x)_i:=D^{-1/2}\sign(\langle x, Z_i\rangle )
\]
where $\varphi^D(x)_i$ is the $i$th coordinate of $\varphi^D(x)$.\footnote{We remark that the map $\varphi^D$ could be referred to as a \emph{random feature map}, see \cite{DBLP:conf/nips/ChoS09,pmlr-v80-liao18a}.}

The maps $\varphi_\ell$ are now defined as the $\ell$-fold composition of maps of the type $\varphi^D$. 
That is, for some integers $D_1,D_2,\dotsc, D_\ell$, we let $\varphi_1:S^{d-1}\to D_1^{-1/2}\{-1,1\}^{D_1}$ be defined by $\varphi_1(x) = \varphi^D(x)$. 
For $j \in \{1,\dotsc, \ell-1\}$, we let $\varphi_{j+1}(x) = \varphi^{D_{j+1}}(\varphi_j(x))$. Therefore, the map $\varphi_\ell$ maps $S^{d-1}$ to $D_{\ell}^{-1/2}\{-1,1\}^{D_\ell}$. 
To avoid writing the double subscript we write the final dimension of the map as $D_\ell=N$.

It was shown in \cite{pmlr-v80-liao18a} that for $x,y \in S^{d-1}$, $\e \sign(\langle x,Z_1\rangle)\sign(\langle y,Z_1\rangle)  = f(\langle x,y \rangle )$. 
Since $\langle\varphi^D(x),\varphi^D(y)\rangle$ is a sum of $D$ independent copies of $D^{-1}\sign(\langle x,Z\rangle)\sign(\langle y,Z\rangle)$ we get that  $\e \langle\varphi^D(x),\varphi^D(y)\rangle = f(\langle x,y \rangle) $.

Now we can explain how one recovers pairwise distances between points in $S$ from $\varphi_\ell(S)$. 
Let $S$ be a set of $n$ points in $S^{d-1}\subset \Rl^d$. 
As in \cref{thm:sphere}, we map $S$ by $\varphi_\ell:S^{d-1}\to N^{-1/2}\{-1,1\}^{N}$. 
Here $\ell$ is some parameter which is chosen based on $S$ that is explained below and $N$ is chosen based on the desired $\epsilon$ error of the sketch. 
If the remaining integers $D_1,\dots, D_{\ell-1}$ are chosen properly, then we show that $\langle \varphi_\ell(x),\varphi_\ell(y)\rangle $ is a good approximation of $f_\ell(\langle x,y\rangle)$ (\cref{cor}). 
Since $g_\ell $ is the inverse of $f_\ell$ this implies that $g_\ell(\langle \varphi_\ell(x),\varphi_\ell(y)\rangle)$ should be a good approximation of $\langle x,y\rangle$. 
By the polarization identity, this implies that $2-2 g_\ell(\langle \varphi_\ell(x),\varphi_\ell(y)\rangle)$ should be a good approximation of $\|x-y\|^2$. 
So recovering $\|x-y\|^2$ from $\varphi_\ell(S)$ simply involves calculating $2-2 g_\ell(\langle \varphi_\ell(x),\varphi_\ell(y)\rangle)$. 
In the next section we explain why this mapping and recovery scheme leads to good error guarantees. 

\subsection{Intuition behind the construction}\label{sec:intuition}

Now that we have defined the maps $\varphi_\ell$ we can explain the idea behind using maps of this form. The reason why this type of map is useful has to do with the behavior of the derivative of the function $g_\ell$ near $t=1$. As previously mentioned, the  map $\varphi_\ell$ has the property that for all $x,y \in S$,
\begin{align}\label{eq:errr}
|\langle \varphi_\ell(x),\varphi_\ell(y)\rangle -f_\ell(\langle x,y\rangle)|<\delta
\end{align}
for some $\delta$ depending on $N$. Now when we want to recover $\|x-y\|^2$ based on $\varphi_\ell(x)$ and $\varphi_\ell(y)$ we compute $2-2g_\ell(\langle \varphi_\ell(x),\varphi_\ell(y)\rangle)$. The additive error of the approximation of $\|x-y\|^2$ by $2-2g_\ell(\langle \varphi_\ell(x),\varphi_\ell(y)\rangle)$  depends on the error in \cref{eq:errr} as well as the derivative of $g_\ell$ near the point $f_\ell(\langle x,y\rangle)$. The function $g_\ell$ has the property that its derivative approaches zero as $t$ approaches one and so the additive error of the approximation of $\|x-y\|^2$ by $2-2g_\ell(\langle \varphi_\ell(x),\varphi_\ell(y)\rangle)$  gets smaller the closer $f_\ell(\langle x,y\rangle) $ (and thus $\langle x,y\rangle$) is to one, i.e., the closer $\|x-y\|^2$ is to zero. 
(This is quantified in \cref{thm:additive}, specifically the exponent approaching 2 in the additive error)
The effect that this has is that we actually approximate $\|x-y\|^2$ up to a multiplicative error.

The role of the parameter $\ell$ in this construction is in controlling the rate at which $g_\ell'(t)$ approaches zero as $t$ approaches one. If there are pairs $x,y$ such that $\|x-y\|^2$ is very small then we need the derivative to approach zero very quickly. This can be done by increasing the parameter $\ell$; the rate at which $g_\ell'$ approaches zero is faster for larger values of $\ell$ (see \cref{thm:d}). So if the point set $S$ has pairs of points at a very small distance, we map $S$ by $\varphi_\ell$ for a larger value of $\ell$ in order to decrease the approximation error for these very close pairs. It turns out that the correct choice of $\ell$ is $\lceil \log_2\log_2 r\rceil$ where $r$ is approximately the reciprocal of the minimum distance between pairs of points in $S$ (see \cref{thm:sphere}).

We remark that the function $g_\ell(t)$ has the property that its derivative can be as large as $\approx \frac{\pi}{2}^\ell$  when $t$ is near zero (see \cref{lem:deriv3}). The effect that this has is that our algorithm leads to worse approximation of $\|x-y\|^2$ when $\langle x,y \rangle$ is close to zero if $\ell$ is large. However, this loss in accuracy is made up for by increasing $N$ by only a relatively small amount and the gain in accuracy when $\langle x,y\rangle$ is close to one outweighs the loss in accuracy when $\langle x,y\rangle$ is close to zero.

\subsection{Previous variations on random projection}
Our compression method is similar to random projection in that they both involve compressing a set of points by randomly mapping it to a lower dimensional space. A number of other papers have also suggested variations on random projection where a different random mapping is used. In some of these variations mentioned below, the random mapping is still linear. In others, they use a linear mapping followed by some quantization step.  
%The main difference in our method is that no part of the random mapping is linear.\lnote{discuss/reword}
The main difference in our method is that it is more fundamentally non-linear due to the fact that it is a ``deep'' composition of non-linear maps.

One of the standard versions of random projection involves mapping points by a random Gaussian matrix. It was later shown that other types of random matrices work equally well. In particular, it was shown in \cite{DBLP:conf/pods/Achlioptas01} and \cite{DBLP:journals/ml/ArriagaV06} that the random matrix $R$ in \cref{lem:JL} can be replaced by a matrix whose entries are i.i.d.\ random variables taking value $+1$ with probability $1/2$ and $-1$ with probability $1/2$ while maintaining a similar guarantee on the projected points. This result has been generalized to other types of random matrices, see \cite[Theorem 5.3]{Santosh}.

The above mentioned \say{binary} version of the Johnson-Lindenstrauss lemma due to \cite{DBLP:conf/pods/Achlioptas01} (where the entries of the matrix are all either $+1$ or $-1$) is particularly important for the following reason. As discussed in \cite{MR4418813}, an alternate way to convert sketches obtained by the Johnson-Lindenstrauss lemma to bits is possible if the points have bounded integer coordinates and one uses the binary variant of Johnson-Lindenstrauss lemma. This approach is somewhat incomparable to our setting because of the integrality assumption.

Another variation on the random projection technique is to apply an additional compression step after the random projection. The projected data points are real vectors which in general require an infinite number of bits to store. Therefore, a number of works have proposed a quantization step after the projection which further reduces the cost of storing the data points \cite{DBLP:conf/stoc/Charikar02,DBLP:conf/aistats/Li17,DBLP:conf/nips/LiL19,inproceedings}.

\subsection{Distance compression beyond random mappings}

Random mappings are of course not the only way to compress a data set. Here we compare our method to compression techniques that use methods other than random mappings. These methods tend to be more complicated algorithmically but as we explain below, can produce sketches with fewer bits. 

Given a set of $n$ points in the unit ball in $\Rl^d$, what is the minimum number of bits of a sketch which allows one to recover all pairwise distances up to a multiplicative $(1\pm \epsilon)$ error? As explained above, the Johnson-Lindenstrauss lemma shows that $O\bigl(\epsilon^{-2} n\log n\log(1/m^2\epsilon)\bigr)$ bits suffice. However, this is not the optimal number of bits. It was recently shown in \cite{MR3627774,MR4418813} that if the points are contained in the unit ball and $m$ is the minimum distance between points, then $O\bigl(\epsilon^{-2}n \log n+n \log \log(1/m)\bigr)$ bits suffice. Previous to this result, the best known result was that $O\bigl(\epsilon^{-2}n \log n \log(1/m)\bigr)$ bits suffice \cite{MR1769366}. These two results, however, use sketching techniques that go beyond random mappings. In fact, they use sketches which differ from the sketches obtained by the random projection technique in a fundamental way. Random projection compresses the data set \say{point by point} in the sense that the compression process is applied to each point independently from the others. In contrast, the sketches in \cite{MR4418813} and \cite{MR1769366} must compress the entire data set simultaneously.

%Another way of stating this distinction is that the random projection method satisfies the requirements of the \emph{one-way communication} version of this sketching problem while the methods used in \cite{MR4418813} and \cite{MR1769366} do not. 
Another way of stating this distinction is that ``point by point'' methods (such as random projection) satisfy the requirements of the \emph{one-way communication} version of this sketching problem while the methods used in \cite{MR4418813} and \cite{MR1769366} do not.
In the \emph{one-way communication} version of the sketching problem, Alice holds half of the data points and Bob holds the other half. Alice sends a message to Bob using as few bits as possible. Bob then must report distances between pairs of points where one point in the pair is known by Alice and the other by Bob. The one-way communication version of the sketching problem asks one to determine the minimum number of bits of Alice's message. It was shown in \cite{MR3203011} that if the points are in the unit ball and the minimum distance is $m$, then $\Omega\bigl(\epsilon^{-2}n \log(n/\delta) \log(1/m) \bigr)$ bits are required for the one-way communication version of the problem if the sketch is required to be successful with probability at least $1-\delta$. 

Any sketching algorithm which compresses the data set point by point satisfies the requirements of the one-way communication variant of the sketching problem. We therefore know that sketching algorithms which compress the data set point by point cannot produce sketches with the optimal number of bits. However, there are several advantages to sketching algorithms of this sort. One advantage is that they are generally simpler and easier to implement. Another is that if one wants to add additional points to the data set, the entire sketching algorithm does not need to be re-run, one can simply compress the additional points independently from the rest. 

Our sketching algorithm from \cref{thm:sphere} also has the property that it compresses the data set point by point. Furthermore, the number of bits of our sketch almost matches the lower bound from \cite{MR3203011}. The dominant term in the bound from \cref{thm:sphere} is $\Theta(\epsilon^{-2} n \log n(\log(1/m))^{2\log_2(\pi/\sqrt{2})} )$. Thus our number of bits matches the lower bound from \cite{MR3203011} up to the power on the $\log(1/m)$ term. This motivates the question of whether some variation on our sketching technique can reduce this power. 

\subsection{Outline of the paper}
\cref{sec:construction} contains the construction of the maps $\varphi_\ell$ and quantifies the error in the approximation of $f_\ell(\langle x,y \rangle )$  by  $\langle \varphi_\ell(x),\varphi_\ell(y)\rangle $. Then in \cref{sec:recovery} we prove \cref{thm:sphere,thm:main}. This amounts to showing how $\langle \varphi_\ell(x), \varphi_\ell(y)\rangle $ allows one to estimate $\|x-y\|^2$ and quantifying the error of the estimation. 

\subsection{Notation and preliminaries}
For $x\in \Rl^d \setminus\{ 0 \}$ we use $\hat x $ to denote $x/\|x\|$. For $x,y \in S^{d-1}$, the \emph{polarization identity} states that $2-2\langle x,y \rangle = \|x-y\|^2$; for arbitrary $x,y \in \Rl^d$, it states that $\|x\|^2+\|y\|^2- 2\langle x,y \rangle = \|x-y\|^2$.

\section{The construction of the maps $\varphi_\ell$}\label{sec:construction}

The purpose of this section is to prove the following result, \cref{cor}, which shows a bound on the error of the approximation of $f_\ell(\langle x,y \rangle)$ by $\langle \varphi_\ell(x),\varphi_\ell(y)\rangle$ for all pairs $x,y$ in a set of $n$ points. The error in this approximation depends on the dimension $D_\ell$ of the image space of $\varphi_\ell$. In particular, we show how large $D_\ell$ needs to be in order to guarantee with high probability that  $\langle \varphi_\ell(x),\varphi_\ell(y)\rangle$ is equal to $f_\ell(\langle x,y \rangle)$ up to some additive error $\delta$ for all pairs $x,y$. 

We recall the definitions of the functions $f_\ell$ and $g_\ell$: Let $f(t) = \frac{2}{\pi}\arcsin(t)$ and $g(t) = \sin(\frac{\pi t}{2})$. 
For $\ell\in \mathbb{N}^+$, let $f_\ell$ be the function $f$ composed with itself $\ell$ times and similar for $g_\ell$. 
Notice that for any $\ell \in \mathbb{N}^+$, $f_\ell:[-1,1]\to [-1,1]$ is the inverse of $g_\ell:[-1,1] \to [-1,1]$. 

\begin{corollary}\label{cor}
Let $S \subset S^{d-1}$ with $|S|=n\ge 2$. 
Let $r:= \max_{x,y \in S, x \neq y}\frac{2}{\sqrt{1-|\langle x, y\rangle|}}$. 
Let $\ell \in \mathbb{N}^+$. 
Let $\delta >0$ be such that $\delta < \frac{2}{r^2}$\details{$=\min \frac{1-\abs{\langle x,y\rangle}}{2}$.}. 
Then there exists a random map $\varphi_{\ell}: \Rl^{d} \to D_\ell^{-1/2}\{-1,1\}^{D_\ell}$ (independent of $S$ except through parameters $n, d, r$) such that %$\|\varphi_j(x)\|_2=1$ for all $x \in S$ and $j \in \{1, \dotsc, \ell\}$.\lnote{isn't `such that' part redundant now?}
%Moreover, 
with probability $(1-\frac{1}{n})^\ell$ it satisfies
\[
\bigl| \langle \varphi_{\ell}(x),\varphi_{\ell}(y) \rangle - f_\ell(\langle x,y \rangle ) \bigr| 
< \delta
\]
for all $x,y \in S$, where $D_{\ell} \le \lceil \frac{24\log n}{\delta^2}\rceil$.
\end{corollary}

The above corollary follows immediately from the following theorem which also determines a bound on the error of the approximation accuracy not only of $f_\ell(\langle x,y \rangle)$ by  $\langle \varphi_\ell(x),\varphi_\ell(y)\rangle $, but also the approximation accuracy of $f_j(\langle x,y \rangle)$ by  $\langle \varphi_j(x),\varphi_j(y)\rangle $ for all $j \in \{1,\dotsc, \ell\}$.

We will make use of the fact that for all $j$, $f_j$ and $g_j$ are both increasing on $[-1,1]$ and that $f_j(0) = g_j(0)=0$, $f_j(1) = g_j(1)=1$ and  $f_j(-1) = g_j(-1)=-1$. 
We will also make use of the fact that for any $j \in \mathbb{N}^+$, $f_j$ and $g_j$ are odd functions.

\begin{theorem}\label{thm:conditional2}
Let $S \subset S^{d-1}$ with $|S|=n\ge 2$. 
Let $r:= \max_{x,y \in S, x \neq y}\frac{2}{\sqrt{1-|\langle x, y\rangle|}}$. 
Let $\ell \in \mathbb{N}^+$. 
Let $\delta >0$ be such that $\delta < \frac{2}{r^2}$\details{$=\min \frac{1-\abs{\langle x,y\rangle}}{2}$.}. 
Then there exist random maps $\varphi_{j}: \Rl^{d} \to D_j^{-1/2}\{-1,1\}^{D_j}$, $j \in \{1, \dotsc, \ell\}$ (independent of $S$ except through parameters $n, d, r$) such that %$\|\varphi_j(x)\|_2=1$ for all $x \in S$ and $j \in \{1, \dotsc, \ell\}$.\lnote{isn't `such that' part redundant now?}
%Moreover, 
for all $j \in \{1, \dotsc, \ell\}$, with probability at least $(1-\frac{2}{n})^j$, $\varphi_j$ satisfies
\begin{align}\label{eq:je}
\bigl| \langle \varphi_{j}(x),\varphi_{j}(y) \rangle - f_j(\langle x,y \rangle ) \bigr| 
< \frac{\delta}{2^{\ell-j}r^{3((2/3)^{j}-(2/3)^\ell)}}
\end{align}
for all $x,y \in S$, where $D_{j} \le \Bigl\lceil \frac{24\cdot2^{2(\ell-j)}r^{6((2/3)^{j}-(2/3)^\ell)}\log n}{\delta^2} \Bigr\rceil$.
\end{theorem}

\begin{proof}
The maps $\varphi_j$ will be compositions of maps of the following form. Set $D \in \mathbb{N}^+$. Let $Z_i$, $1\le i\le D$ be i.i.d.\ standard Gaussian random vectors in $\Rl^d$. Let $\varphi^D: \Rl^d\to \Rl^D$ be defined by 
\[
\varphi^D(x)_i:=D^{-1/2}\sign(\langle x, Z_i\rangle ),
\]
where $\varphi^D(x)_i$ is the $i$th coordinate of $\varphi^D(x)$. A direct calculation, see \cite{pmlr-v80-liao18a}, shows that $\e \langle \varphi^D(x),\varphi^D(y)\rangle = f(\langle x,y \rangle)$. Furthermore, $\|\varphi^D(x)\|=1$ for all $x \in \Rl^d$.%\lnote{minor technicality: a.s. and need to define sign to be -1 1 to avoid issues with statement of theorem} 

First we use this construction to define $\varphi_1$. We let $\varphi_1 = \varphi^{D_1}$ where $D_1$ is chosen below. Using Hoeffding's inequality, for all $x,y \in S$,
\begin{align*}
   & \prob\left(\Big| \langle \varphi(x) , \varphi(y)\rangle - f(\langle x,y \rangle)\Big|  >\frac{\delta}{2^{\ell-1}r^{3((2/3)-(2/3)^\ell)}}\right) \\
    &= \prob\left(\Big| D_1\langle \varphi(x) , \varphi(y)\rangle - D_1f(\langle x,y \rangle)\Big|  > \frac{D_1\delta}{2^{\ell-1}r^{3((2/3)-(2/3)^\ell)}}\right) \\
    & \le 2 \exp\left(-\frac{D_1\delta^2}{2\cdot2^{2(\ell-1)}r^{6((2/3)-(2/3)^\ell)}}\right).
\end{align*}
We set $D_1 = \Bigl\lceil \frac{6\cdot2^{2(\ell-1)}r^{6((2/3)-(2/3)^\ell)}\log n}{\delta^2}\Bigr\rceil$. This means that the above probability is less than $2/n^3$ and that $\varphi_1$ satisfies the conditions of the theorem with probability at least $1-\binom{n}{2}\frac{2}{n^3}\ge 1-1/n\ge 1-2/n$. 

Now assume that the required map exists for some $j \ge 1$. We  will show that it exists for $j+1$. So there exists a map $\varphi_{j}: \Rl^{d} \to D_j^{-1/2}\{-1,1\}^{D_j}$ which satisfies \cref{eq:je} with probability at least $(1-\frac{2}{n})^j$. Let $\varphi_{j+1}$ be defined by $\varphi_{j+1}(x) = \varphi^{D_{j+1}}\bigl(\varphi_j(x)\bigr)$ where $D_{j+1}$ will be chosen at the end of the proof. 
We will show that, conditioned on the event that $\varphi_j$ does satisfy \cref{eq:je}, the probability that $\varphi_{j+1}$ satisfies \cref{eq:je} is at least $1-\frac{2}{n}$. So assume that $\varphi_j$ does satisfy \cref{eq:je}. Recall that we want to show that $\langle \varphi_{j+1}(x), \varphi_{j+1}(y)\rangle$ is a good estimate of $f_{j+1}(\langle x,y \rangle)$. We have by \cite{pmlr-v80-liao18a} that 
\[
\e \langle \varphi_{j+1}(x), \varphi_{j+1}(y)\rangle = f(\langle \varphi_j(x),\varphi_j(y) \rangle),
\]
i.e., $\langle \varphi_{j+1}(x), \varphi_{j+1}(y)\rangle$ is an unbiased estimator of $f(\langle \varphi_j(x),\varphi_j(y) \rangle)$. 
By the triangle inequality,
\begin{align}\label{eq:tri}
&\bigl|\langle \varphi_{j+1}(x), \varphi_{j+1}(y)\rangle - f_{j+1}(\langle x,y \rangle) \bigr| \notag\\ 
&\le \bigl|\langle \varphi_{j+1}(x), \varphi_{j+1}(y)\rangle -f(\langle \varphi_j(x),\varphi_j(y) \rangle )\bigr|  + \bigl|f(\langle \varphi_j(x),\varphi_j(y) \rangle )  - f(f_j(\langle x,y \rangle))\bigr|.
\end{align}
First we give a bound on the second term in \cref{eq:tri}. We are assuming that $\varphi_j$ satisfies \cref{eq:je}, i.e., that for all $x,y\in  S$,
\[
 \bigl|\langle \varphi_j(x),\varphi_j(y) \rangle  - f_j(\langle x,y \rangle)\bigr| \le \frac{\delta}{2^{\ell-j}r^{3((2/3)^{j}-(2/3)^\ell)}}.
\]
So to get a bound on the second term in \cref{eq:tri} we need to get an upper bound on the derivative of $f$ in the interval between $\langle \varphi_j(x),\varphi_j(y) \rangle$ and $f_j(\langle x,y \rangle)$. We claim for all pairs $x,y\in S$, $x\neq y$, the derivative of $f$ in the interval between $\langle \varphi_j(x),\varphi_j(y) \rangle$ and $f_j(\langle x,y \rangle)$ is upper bounded by $r^{(2/3)^j}$. 

Let $t$ be in the closed interval between $\langle \varphi_j(x),\varphi_j(y) \rangle$ and $f_j(\langle x,y \rangle)$.
By definition of $\varphi_j$ and $f_j$ this implies $\abs{t} \leq 1$.
We also have $\abs{t} \leq \abs{f_j(\langle x,y \rangle)} + \delta/2^{\ell-j}r^{3((2/3)^{j}-(2/3)^\ell)} \le \abs{f_j(\langle x,y \rangle)} + \delta$ by \cref{eq:je} and the fact that $ r \geq 2$ by definition. 
Thus, 
\begin{align*}
    1-|t| &\ge  1-|f_j(\langle x,y \rangle)|-\delta \\
    & \ge (1-|\langle x,y\rangle|)^{(2/3)^j}-\delta && \text{by \cref{lem:lower}} \\
    & \ge \frac{(1-|\langle x,y\rangle|)^{(2/3)^j}}{2} && \text{by } \delta<\frac{2}{r^2}\le \frac{(1-|\langle x,y \rangle|)}{2}\le \frac{(1-|\langle x,y \rangle|)^{(2/3)^j}}{2}.
\end{align*}
Using this and $\abs{t}\leq 1$, we get
\begin{align*}
f'(t) &= \frac{2}{\pi \sqrt{1-t^2}}\\
&\leq \frac{2}{\pi \sqrt{1-\abs{t}}}\\
&\leq \frac{2\sqrt{2}}{\pi \bigl(\sqrt{1-|\langle x,y \rangle|}\bigr)^{(2/3)^j}}\\
&\le\frac{2\sqrt{2}r^{(2/3)^j}}{\pi} <r^{(2/3)^j}.
\end{align*}

We have shown that the derivative of $f$ in the interval between $\langle \varphi_j(x),\varphi_j(y) \rangle$ and $f_j(\langle x,y \rangle)$ is upper bounded by $r^{(2/3)^j}$. 
This combined with \cref{eq:je} and the fact that the derivative of $f$ is positive implies that
\begin{align*}
    \Big|f(\langle \varphi_j(x),\varphi_j(y) \rangle )  - f(f_j(\langle x,y \rangle))\Big|& \le r^{(2/3)^j}\frac{\delta}{2^{\ell-j}r^{3((2/3)^j-(2/3)^\ell)}}\\
    &=\frac{\delta}{2\cdot2^{\ell-(j+1)}r^{3((2/3)^{j+1}-(2/3)^\ell)}},
\end{align*}
where the equality above uses that $3 \bigl((2/3)^j-(2/3)^\ell \bigr)= \sum_{i=j}^{\ell-1}(2/3)^j$.

Now we deal with the first term in \cref{eq:tri}. Using Hoeffding's inequality,
\begin{align*}
    &\prob\Bigl(\bigl| \langle \varphi_{j+1}(x) , \varphi_{j+1}(y)\rangle - f_{j+1}(\langle x,y \rangle)\bigr|  >\frac{\delta}{2\cdot2^{\ell-(j+1)}r^{3((2/3)^{j+1}-(2/3)^\ell)}}\Bigr) \\
    &=\prob\Bigl(\bigl| D_{j+1}\langle \varphi_{j+1}(x) , \varphi_{j+1}(y)\rangle - D_{j+1}f_{j+1}(\langle x,y \rangle)\bigr|  > \frac{D_{j+1}\delta}{2\cdot2^{\ell-(j+1)}r^{3((2/3)^{j+1}-(2/3)^\ell)}}\Bigr)\\
    & \le 2\exp\Bigl(-\frac{\delta^2D_{j+1}}{8\cdot 2^{2(\ell-(j+1))}r^{6((2/3)^{j+1}-(2/3)^\ell)}}\Bigr).
\end{align*}
We set $D_{j+1} = \lceil \frac{24\cdot2^{2(\ell-(j+1))}r^{6((2/3)^{j+1}-(2/3)^\ell)}\log n}{\delta^2}\rceil$. This means that the above probability is less than $2/n^3$. So, using \cref{eq:tri} and the previously established bound on the second term in \cref{eq:tri}, we have shown that for any pair $x,y \in S$,
\[
\Big| \langle \varphi_{j+1}(x) , \varphi_{j+1}(y)\rangle - f_{j+1}(\langle x,y \rangle)\Big|  <\frac{\delta}{2^{\ell-(j+1)}r^{3((2/3)^{j+1}-(2/3)^\ell)}}
\]
with probability at least $1-\frac{2}{n^3}$. 
So, conditioned on the event that $\varphi_j$ satisfies \cref{eq:je}, $\varphi_{j+1}$ satisfies \cref{eq:je} with probability at least $1-\frac{2}{n}$. 
Since the probability that $\varphi_{j+1}$ satisfies \cref{eq:je} is greater than or equal to the probability that both $\varphi_{j+1}$ and $\varphi_j$ satisfy \cref{eq:je}, this means that the probability that $\varphi_{j+1}$ satisfies \cref{eq:je} is at least $(1-\frac{2}{n})^{j+1}$.
\end{proof}

\section{The recovery of $\|x-y\|^2$ by $\varphi_\ell(x)$ and $\varphi_\ell(y)$}\label{sec:recovery}

Here we complete the proofs of the main theorems, \cref{thm:sphere,thm:main}. This is done in two steps. Recall that we showed in \cref{cor} that $\langle \varphi_\ell(x),\varphi_\ell(y)\rangle $ is a good approximation of $f_\ell(\langle x,y \rangle)$ for all pairs $x,y \in S$. The first step is \cref{thm:additive} which shows that $g_\ell(\langle \varphi_\ell(x),\varphi_\ell(y)\rangle )$ is a good approximation of $\langle x,y \rangle$. The reason this is true is because $g_\ell$ is the inverse of $f_\ell$. So the proof of \cref{thm:additive} uses facts about the derivative of $g_\ell$ (in particular \cref{thm:d}) to show that the bound on the error of the approximation of $f_\ell(\langle x,y\rangle)$ by $\langle \varphi_\ell(x),\varphi_\ell(y)\rangle$ implies a bound on the error of the approximation of $\langle x,y \rangle $ by $g_\ell(\langle \varphi_\ell(x),\varphi_\ell(y)\rangle )$ (\cref{thm:additive}). 

The second step is to set $\ell = \lceil\log_2 \log_2 \frac{4}{m} \rceil$ where $m$ is the minimum distance between pairs of distinct points in $S$ and then show using the polarization identity that the error bound established in \cref{thm:additive} implies the error bounds in \cref{thm:sphere,thm:main}.

\begin{theorem}\label{thm:additive}
    Let $S \subset S^{d-1}$ with $|S|=n$. Let $\ell \in \mathbb{N}^+$ and $\epsilon >0$ and assume that $\epsilon$ satisfies $\epsilon<1-|\langle x,y \rangle|$ for all $x,y \in S$ with $x\neq y$. 
    Then the random map $\varphi_\ell : S^{d-1} \to N^{-1/2}\{-1,1\}^N$ from \cref{thm:conditional2} satisfies that, with probability at least $(1-2/n)^\ell$, for all $x,y \in S$, $g_\ell(\langle \varphi_\ell(x),\varphi_\ell(y)\rangle)$ is equal to $\langle x,y \rangle$ up to an additive $\pm \epsilon\|x-y\|^{2-2^{-\ell+1}}$ error where $N= \left\lceil \frac{48(\pi/\sqrt{2})^{2\ell}\log n}{\epsilon^2}\right\rceil$. Equivalently, for all $x,y \in S$, $2-2g_\ell(\langle \varphi_\ell(x),\varphi_\ell(y)\rangle)$ is equal to $\|x-y\|^2$ up to an additive $\pm \epsilon\|x-y\|^{2-2^{-\ell+1}}$ error. \details{Note that $2-2|\langle x,y \rangle| \leq 2-2\langle x,y \rangle = \norms{x-y}$.}
\end{theorem}
\begin{proof}
We will actually prove the stronger result that with probability at least $(1-2/n)^\ell$, for all $x,y \in S$, $g_\ell(\langle \varphi_\ell(x),\varphi_\ell(y)\rangle)$ is equal to $\langle x,y \rangle$ up to an additive $\pm \epsilon(2-2|\langle x,y \rangle|)^{1-2^{-\ell}}$ error. This statement implies the theorem statement by the polarization identity. 

Set $\delta = (\epsilon/\sqrt{2})(\sqrt{2}/\pi)^\ell$ in \cref{thm:conditional2}. 
One can check that the assumption that $\epsilon < 1-|\langle x,y \rangle|$ for all $x,y \in S$ with $x\neq y$ implies that $\delta<2/r^2$ for all $x,y \in S$ with $x\neq y$ as required by \cref{thm:conditional2}. 
From \cref{thm:conditional2}, the map $\varphi_\ell:\Rl^d \to \frac{1}{\sqrt{N}}\{-1,1\}^{N}$  satisfies, with probability at least $(1-2/n)^\ell$, that
\begin{align}\label{eq:er}
\bigl\lvert \langle \varphi_{\ell}(x),\varphi_{\ell}(y) \rangle - f_{\ell}(\langle x,y \rangle ) \bigr\rvert < (\epsilon/\sqrt{2})(\sqrt{2}/\pi)^\ell
\end{align}
for all $x,y \in S$ where $N= \left\lceil \frac{48(\pi/\sqrt{2})^{2\ell}\log n}{\epsilon^2}\right\rceil$.
%\lnote{where do 12 here and $\sqrt{2}$ in (3) come from?}\bnote{The definition of $\delta$ was wrong, that was probably why you were confused. And had to change $\delta $ by a factor of 2 due to the mistake below about the missing factor of $2^{1-2^{-\ell}}$. Fixed it now. The def. of $\delta$ is just chosen so that when you multiply by the bound on $g_\ell'(t)$ the end result is an error bound of $\epsilon\bigl(2-2|\langle x,y \rangle|\bigr)^{1-2^{-\ell}}$. And it is now 48 instead of 12 and the 48 comes from plugging $\delta = (\epsilon/\sqrt{2})(\sqrt{2}/\pi)^\ell$ into Theorem 1 (or the corollary.)}
Now assume that $\varphi_\ell$ does satisfy \cref{eq:er} for all $x,y \in S$. 
Since $g_{\ell}\bigl(f_\ell(t)\bigr) = t $ for all $t \in [-1,1]$, \cref{eq:er} implies that $g_\ell(\langle \varphi_\ell(x),\varphi_\ell(y)\rangle)$ should be a good approximation of $\langle x,y\rangle$. 
In particular, we claim that for all $x,y\in S$, $g_{\ell}(\langle \varphi_\ell(x),\varphi_\ell(y)\rangle)$ is equal to $\langle x,y \rangle $ up to an additive error of $\pm \epsilon(2-2|\langle x,y \rangle|)^{1-2^{-\ell}}$.
In order to show this we first need to get a bound on the derivative of $g_\ell$ in the interval between $\langle \varphi_{\ell}(x),\varphi_{\ell}(y) \rangle$ and $f_{\ell}(\langle x,y \rangle)$.

Let $t$ be in the interval between $\langle \varphi_{\ell}(x),\varphi_{\ell}(y) \rangle$ and $f_{\ell}(\langle x,y \rangle)$. 
By \cref{thm:d}, 
\[
g_\ell'(t)\le \frac{\pi^\ell}{2^{\frac{\ell+1}{2}}}\bigl(2-2g_\ell(|t|)\bigr)^{1-2^{-\ell}}.
\]
By \cref{eq:er}, we have that $ \bigl\lvert \abs{f_\ell(\langle x,y \rangle)}-\abs{t} \bigr\rvert \le \sqrt{2}\epsilon(\frac{\sqrt{2}}{\pi})^\ell$. 
By \cref{lem:deriv3} this means that $ g_\ell\bigl(|f_\ell(\langle x,y \rangle)|\bigr) - g_\ell(|t|)  \le \sqrt{2}\epsilon(\sqrt{2}/\pi)^\ell(\pi/2)^\ell \le \epsilon $, i.e. that $g_\ell(|t|) \ge g_\ell\bigl(|f_\ell(\langle x,y \rangle)|\bigr)-\epsilon$. 
Since $f$ is an odd function, $f_\ell(t)$ is also an odd function and so $g_\ell\bigl(|f_\ell(\langle x,y \rangle)|\bigr) = g_\ell\bigl(f_\ell(|\langle x,y\rangle|)\bigr) =|\langle x,y \rangle|$. 
So we have shown that $g_\ell(|t|) \ge |\langle x,y \rangle|-\epsilon$. 
Since $\epsilon<1-|\langle x,y \rangle|$, we have that $2-2g_\ell(|t|) \le 2-2|\langle x,y \rangle|+2\epsilon <2(2-2|\langle x,y \rangle|)$. 
Using that $2^{1-2^{-\ell}}\le 2$, this means that 
\[
g_\ell'(t) 
\le  \frac{\pi^\ell}{2^{\frac{\ell+1}{2}}}\bigl(2-2g_\ell(|t|)\bigr)^{1-2^{-\ell}} 
\le \sqrt{2}\left(\frac{\pi}{\sqrt{2}}\right)^\ell(2-2|\langle x,y \rangle|)^{1-2^{-\ell}}.
\]
This bound on the derivative along with the fact that $g_\ell'(t)>0$ and \cref{eq:er} means that 
\[
\Bigl\lvert g_\ell\bigl( \langle \varphi_{\ell}(x),\varphi_{\ell}(y) \rangle\bigr) - g_\ell\bigl(f_{\ell}(\langle x,y \rangle )\bigr) \Bigr\rvert < \epsilon\bigl(2-2|\langle x,y \rangle|\bigr)^{1-2^{-\ell}}.
\]
Since $g_\ell\bigl(f_{\ell}(\langle x,y \rangle )\bigr) = \langle x,y \rangle$, the theorem follows.
\end{proof}

Now that we have established the above result, we can prove \cref{thm:sphere,thm:main}. We restate the theorems from the intro for the sake of readability. 

% \begin{theorem}\label{thm:sphere}
% Let $S \subset S^{d-1}$ with $|S|=n\ge 2$. Let $m = \min_{x,y \in S,x\neq y}\| x- y\|$ and let $\ell = \lceil\log_2\log_2(4/m) \rceil \ge 1$. Let $\epsilon>0$ and assume that $\epsilon<1-|\langle x,y \rangle|$ for all $x,y\in S$, $x\neq y$. Then the random map $\varphi_\ell:S^{d-1 } \to N^{-1/2}\{-1,1\}^N$ with $N = \Theta\Big(\epsilon^{-2} (\log(1/m))^{2\log_2(\pi/\sqrt{2})}\log n\Big)$ (defined in the proof of \cref{thm:conditional2}) satisfies the following with probability at least $(1-2/n)^\ell$: $\varphi_\ell( S)$ is a sketch of $S$ that allows one to recover all squared distances between pairs of points in $S$ up to a multiplicative $(1\pm \epsilon)$ error. The number of bits of the sketch is $\Theta\Big(n\epsilon^{-2} (\log(1/m))^{2\log_2(\pi/\sqrt{2})}\log n)\Big)$. 
% \end{theorem}
\sphere*
\begin{proof}
%$g_\ell(\langle \varphi_\ell(x),\varphi_\ell(y)\rangle)$ is equal to $\langle x,y \rangle$ up to an additive $\pm (\epsilon/4)(2-2|\langle x,y \rangle|)^{1-2^{-\ell}}$ error for all $x,y \in S$ with $N= \left\lceil \frac{384(\pi/\sqrt{2})^{2\ell}\log n}{\epsilon^2}\right\rceil$. Notice that this implies that  $g_\ell(\langle \varphi_\ell(x),\varphi_\ell(y)\rangle)$ is equal to $\langle x,y \rangle$ up to an additive $\pm (\epsilon/4)(2-2\langle x,y \rangle)^{1-2^{-\ell}}$ error for all $x,y \in S$. 
Let $\varphi_\ell:S^{d-1} \to N^{-1/2}\{-1,1\}^N$ be the map from \cref{thm:conditional2} that by \cref{thm:additive} satisfies, with probability at least $(1-2/n)^\ell$, that $2-2g_\ell(\langle \varphi_\ell(x),\varphi_\ell(y)\rangle)$ is equal to $\| x-y \|^2$ up to an additive $\pm (\epsilon/4) \|x-y\|^{2-2^{-\ell+1}}$ error for all $x,y \in S$ with $N= \left\lceil \frac{384(\pi/\sqrt{2})^{2\ell}\log n}{\epsilon^2}\right\rceil$. Now assume that $\varphi_\ell$ does satisfy this condition for all $x,y \in S$. 
We have 
\[
    \|x-y\|^{-2^{-\ell+1}}
    \le (1/m)^{2^{-\ell+1}}
    = \bigl((1/m)^{2^{-\ell}}\bigr)^2 
    \le \bigl((1/m)^{\frac{1}{\log_2(4/m)}}\bigr)^2 <4.
\]
This means that we actually estimate $\|x-y\|^2$ up to an additive $\pm \epsilon \|x-y\|^2$ error, i.e., a multiplicative $(1\pm\epsilon)$ error. We have 
\[
(\pi/\sqrt{2})^{2\ell} \le (\pi/\sqrt{2})^{2(1+\log_2 \log_2(4/m))} = (\pi^2/2)2^{2\log_2(\pi/\sqrt{2})\log_2 \log_2(4/m)} =(\pi^2/2)(\log_2(4/m))^{2\log_2(\pi/\sqrt{2})}
\]
so the result follows.
\end{proof}
The above theorem shows that given a set of points $S\subset \Rl^d$, there exists an appropriate choice of $\ell$ and $N$ so that the random map $\varphi_\ell:S^{d-1}\to N^{-1/2}\{-1,1\}^N$ satisfies, with high probability, that $\varphi_\ell(S)$ is a sketch of $S$ that allows one to recover all squared distances between points in $S$ up to a multiplicative $(1\pm\epsilon)$ error. The next theorem shows that this same sketching algorithm also works for point sets that do not necessarily consist of unit norm points provided that the sketch also stores the approximate norms of points in $S$. 

%\lnote{odd issue: problem is invariant under translation but small norm points are hard for our algorithm}\bnote{because of the term $\rho$, small norm points make the number of bits larger. But not sure if this is what you meant.}\lnote{The sketching problem is not harder if one of the points is 0, but our algorithm would have infinite bit complexity as written.}
% \begin{theorem}\label{thm:main}
%     Let $S\subset B^d$ with $|S|=n\ge 2$.
%      Let $\rho = \min_{x \in S}\|x\|^2$ and $m= \min_{x,y \in S,x\neq y}\|\hat x-\hat y\|$.
%     Let $\epsilon>0$ and assume that $\epsilon<1-|\langle \hat x,\hat y \rangle|$ for all $x,y \in S$, $x\neq y$. 
%    Then the random map $\varphi_\ell:S^{d-1 } \to N^{-1/2}\{-1,1\}^N$ with $N = \Theta\Big(\epsilon^{-2} (\log(1/m))^{2\log_2(\pi/\sqrt{2})}\log n\Big)$ (defined in the proof of \cref{thm:conditional2}) satisfies the following with probability at least $(1-2/n)^\ell$:
%     $\varphi_\ell(\hat S)$ and the norm of each point in $S$ up to an additive $\pm \rho m^2\epsilon/48 $ error is a sketch of $S$ that allows one to recover all squared distances between pairs of points in $S$ up to a multiplicative $(1\pm \epsilon)$ error.
%     Moreover the number of bits of the sketch is $\Theta\Big(n\epsilon^{-2} (\log(1/m))^{2\log_2(\pi/\sqrt{2})}\log n+  n\log(48/\rho m^2\epsilon)\Big)$. 
% \end{theorem}
\main*
\begin{proof}
For each $x \in S$, let $n_x$ be an approximation of $\|x\|$ up to an additive $\pm \rho m^2\epsilon/48 $ error.

First we claim that in order to recover squared distances up to a multiplicative $(1\pm \epsilon)$ error, it suffices to recover squared distances to an additive $\pm\epsilon\|x\|\|y\|\|\hat x -\hat y\|^2$ error. The reason is that for any $x,y \in \Rl^d$, we can prove the inequality $\epsilon\|x\|\|y\|\|\hat x -\hat y\|^2 \le \epsilon\|x-y\|^2$ by observing that $\|x\|\|y\|\|\hat x-\hat y\|^2 = \|x\|\|y\|(2-2\langle\hat x, \hat y \rangle ) =2\|x\|\|y\| - 2\langle x, y \rangle\le \|x\|^2+\|y\|^2-2\langle x, y \rangle = \|x-y\|^2 $.  

Now let $\hat S = \{\hat x \suchthat x \in S\}$. Let $\varphi_\ell:S^{d-1} \to N^{-1/2}\{-1,1\}^N$ with $N= \Theta\Big(\frac{(\pi/\sqrt{2})^{2\ell}\log n}{\epsilon^2}\Big)$ be the map from \cref{thm:conditional2} that by \cref{thm:additive} satisfies, with probability at least $(1-2/n)^\ell$, that $g_\ell(\langle \varphi_\ell(\hat x),\varphi_\ell(\hat y)\rangle)$ is equal to $\langle \hat x,\hat y \rangle$ up to an additive $\pm (\epsilon/32)(2-2|\langle \hat x,\hat y \rangle|)^{1-2^{-\ell}}$ error for all $\hat x,\hat y \in \hat S$. Assume that $\varphi_\ell$ does satisfy this condition for all $\hat x,\hat y \in \hat S$.
Notice that this implies that  $g_\ell(\langle \varphi_\ell(\hat x),\varphi_\ell(\hat y)\rangle)$ is equal to $\langle \hat x,\hat y \rangle$ up to an additive $\pm (\epsilon/32)(2-2\langle \hat x,\hat y \rangle)^{1-2^{-\ell}}$ error for all $\hat x,\hat y \in \hat S$. 
We have
\[
(2-2\langle \hat x,\hat y \rangle)^{-2^{-\ell}} 
\le (1/m)^{2^{-\ell+1}}
= \bigl((1/m)^{2^{-\ell}}\bigr)^2 
\le \bigl((1/m)^{\frac{1}{\log_2(4/m)}}\bigr)^2 
< 4.
\]
This means that $g_\ell(\langle \varphi_\ell(\hat x),\varphi_\ell(\hat y)\rangle)$ is equal to $\langle \hat x,\hat y \rangle$ up to an additive $\pm (\epsilon/8)(2-2\langle \hat x,\hat y \rangle)$ error for all $\hat x,\hat y \in \hat S$, i.e., an  additive $\pm (\epsilon/8)\| \hat x-\hat y\|^2$ error. For any $x,y\in S$, 
\begin{align*}
n_xn_yg_\ell(\langle \varphi_\ell(\hat x),\varphi_\ell(\hat y)\rangle) 
&\le \bigl(\|x\|+\rho m^2\epsilon/48\bigr)\bigl(\|y\|+\rho m^2\epsilon/48\bigr)\bigl(\langle\hat x, \hat y \rangle +(\epsilon/8)\| \hat x-\hat y\|^2\bigr) \\
& \le \langle x,y \rangle + (\epsilon/4)\|x\|\|y\|\|\hat x-\hat y\|^2
\end{align*}
where the second inequality uses the definition of $\rho$ and $m$.\details{we have 8 terms, one gives $\langle x,y \rangle$ and one gives $(\epsilon/8)\|x\|\|y\|\| \hat x-\hat y\|^2$. The other six terms all involve at least one factor of $\rho$ so follows by def. of $\rho$.}  We can also show that 
\[n_xn_yg_\ell(\langle \varphi_\ell(\hat x),\varphi_\ell(\hat y)\rangle) \ge \langle x,y \rangle - (\epsilon/4)\|x\|\|y\|\|\hat x-\hat y\|^2.
\]
This means that $n_xn_yg_\ell(\langle \varphi_\ell(\hat x),\varphi_\ell(\hat y)\rangle)$ approximates $\langle x,y\rangle$ up to an additive $\pm(\epsilon/4)\|x\|\|y\|\|\hat x-\hat y\|^2$ error. Since $n_x$ approximates $\|x\|$ up to an additive $\pm(\epsilon/24)\min_{x,y \in S}\|x\|\|y\|\|\hat x - \hat y\|^2$ error this means that $n_x^2$ approximates $\|x\|^2$ up to at least an additive $\pm(\epsilon/4)\|x\|\|y\|\|\hat x - \hat y\|^2$ error. 
Now this means that 
\[
n_x^2+n_y^2-2n_xn_yg_\ell(\langle \varphi_\ell(\hat x),\varphi_\ell(\hat y)\rangle)
\]
approximates 
\[
\|x\|^2+\|y\|^2-2\langle x,y \rangle = \|x-y\|^2
\]
up to an additive $\pm\epsilon\|x\|\|y\|\|\hat x - \hat y\|^2$ error.

Storing all the norms of the points up to an additive $\pm \rho$ error requires $\log(1/\rho)$ bits per point.%\bnote{explain net. edit: no net needed}%\lnote{doesn't one need to store $\rho$? but it seems easy, can store to within factor 2 with $\log_2 \log_2 (1/\rho)$ bits.}

We have 
\[
(\pi/\sqrt{2})^{2\ell} \le (\pi/\sqrt{2})^{2(1+\log_2 \log_2 r)} = (\pi^2/2)2^{2\log_2(\pi/\sqrt{2})\log_2 \log_2 r} =(\pi^2/2)(\log_2 r)^{2\log_2(\pi/\sqrt{2})}
\]
so the result follows.
\end{proof}

\appendix
\section{Technical lemmas}

\begin{theorem}\label{thm:d}
    For all $t \in [0,1]$, 
    \[
    g_\ell'\bigl(f_\ell(t)\bigr) \le \frac{\pi^\ell}{2^{\frac{\ell+1}{2}}}(2-2t)^{1-2^{-\ell}}.
    \]
    This implies that for all $t \in [0,1]$,
      \[
    g_\ell'(t) \le \frac{\pi^\ell}{2^{\frac{\ell+1}{2}}} \bigl(2-2g_\ell(t)\bigr)^{1-2^{-\ell}}.
    \]
    and that for all $t \in[-1,1]$,
    \[
    g_\ell'(t) \le \frac{\pi^\ell}{2^{\frac{\ell+1}{2}}} \bigl(2-2g_\ell(|t|)\bigr)^{1-2^{-\ell}}.
    \]
\end{theorem}
\begin{proof}
We start by proving the first claim.

We will use the inequality
\begin{align}\label{eq:comp}
1-f(t)=\frac{2}{\pi}\arccos(t)\le\sqrt{1-t} \text{  for } t\in[ 0,1],
\end{align}
%which can be checked by computer.\bnote{can also prove it by finding critical points of $\sqrt{1-t}-\frac{2}{\pi}\arccos(t)$. But maybe there is a simpler proof}\lnote{discuss} 
which follows by finding critical points of $\sqrt{1-t}-\frac{2}{\pi}\arccos(t)$.
Now this enables us to prove by induction that 
\begin{align*}
1-f_{\ell}(t)\le (2-2t)^{2^{-\ell}}  \text{  for } t\in[ 0,1]    
\end{align*}
for all $\ell \in \mathbb{N}^+$.
The base case $\ell=1$ is (using \cref{eq:comp})
\[
1-f(t) = \frac{2}{\pi}\arccos(t)\le\sqrt{1-t} \le \sqrt{2-2t}.
\]
The induction step again uses \cref{eq:comp} and also uses that $f_{\ell-1}(t) \in [0,1]$ if $t \in [0,1]$. For all $t \in [0,1]$, we have
\begin{align*}
    1-f_\ell(t) &= 1-f(f_{\ell-1}(t))\\
    &\le \sqrt{1-f_{\ell-1}(t)}\\
    & \le \sqrt{(2-2t)^{2^{-\ell+1}}}\\
    &=(2-2t)^{2^{-\ell}}.
\end{align*}
So $1-f_{\ell}(t)\le (2-2t)^{2^{-\ell}}$ for $t\in[ 0,1]$ follows. 

Now we will prove the first theorem claim by induction. 
The base case is
\[
g'\bigl(f(t)\bigr) = \frac{\pi}{2}\cos\bigl(\arcsin(t)\bigr) = \frac{\pi}{2}\sqrt{1-t^2} \le \frac{\pi}{2}\sqrt{2-2t}.
\]
Now for the induction step, assume that the first theorem claim holds for $f_{\ell-1}$. We first need to establish that for all $t \in [-1,1]$,
\begin{align*}
    g'\bigl(f_\ell(t)\bigr) &= \frac{\pi}{2}\cos\bigl((\pi/2)f_\ell(t)\bigr)\\
    &= \frac{\pi}{2}\cos\bigl(\arcsin(f_{\ell-1}(t))\bigr)\\
    &= \frac{\pi}{2}\sqrt{1-(f_{\ell-1}(t))^2}\\
    & \le \frac{\pi}{\sqrt{2}}\sqrt{1-f_{\ell-1}(t)} && \text{ by } f_{\ell-1}(t) \le 1\\
    & \le \frac{\pi}{\sqrt{2}}(2-2t)^{2^{-\ell}}.
\end{align*}
Using the chain rule, $g_\ell'(t) = g'(t)g_{\ell-1}'\bigl(g(t)\bigr)$. So for all $t \in [0,1]$,
    \begin{align*}
      g_\ell'\bigl(f_\ell(t)\bigr) &= g'\bigl(f_\ell(t)\bigr)g_{\ell-1}'\bigl(g(f_\ell(t))\bigr)\\
      &=  g'\bigl(f_\ell(t)\bigr)g_{\ell-1}'\bigl(f_{\ell-1}(t)\bigr)\\
      &\le \frac{\pi}{\sqrt{2}}(2-2t)^{2^{-\ell}}\frac{\pi^{\ell-1}}{2^{\frac{\ell}{2}}}(2-2t)^{1-2^{-\ell+1}}\\
      &= \frac{\pi^\ell}{2^{\frac{\ell+1}{2}}}(2-2t)^{1-2^{-\ell}}. 
    \end{align*}
The second claim follows by plugging in $g_\ell(t)$ into the first claim. The third claim follows by observing that if $t \in [-1,0]$, then $g_\ell'(t) = g_\ell'(|t|)$ since $g_\ell'$ is an even function (by the fact that $g_\ell$ is odd.)
\end{proof}

\begin{lemma}\label{lem:deriv3}
    For all $t \in [-1,1]$, $0\le g_\ell'(t) \le (\pi/2)^\ell$.
\end{lemma}
\begin{proof}
    By induction. 
    When $\ell=1$, $g'(t) = \frac{\pi}{2}\cos(\frac{\pi t}{2})\in [0,\pi/2]$ when $t\in [-1,1]$. 
    For the induction step, we need to use the fact that $g_{\ell-1}(t) \in [-1,1]$ when $t\in [-1,1]$.
    Using this, 
    \[
    g_\ell'(t) = g'\bigl(g_{\ell-1}(t)\bigr) g_{\ell-1}'(t) =\frac{\pi}{2}\cos\left(\frac{\pi}{2}g_{\ell-1}(t)\right)g_{\ell-1}'(t) \in [0,(\pi/2)^\ell].\qedhere
    \]
\end{proof}

\begin{lemma}\label{lem:upper}
    $|f_\ell(t)|\le |t|$ for all $t \in [-1,1]$ and all $\ell \in \NN^+$.
\end{lemma}
\begin{proof}
  Since $f_\ell(-t) = -f_\ell(t)$ for all $t \in [-1,1]$, it suffices to prove that $f_\ell(t) \le t$ for all $t \in [0,1]$. 
  The claim follows by induction in $\ell$. 
\end{proof}

\begin{lemma}\label{lem:lower}
For all $\ell \in \mathbb{N} $ and $t \in[-1,1]$, $1-|f_\ell(t)|\ge (1-|t|)^{(2/3)^\ell}$.
\end{lemma}
\begin{proof}
Since $f_\ell(-t) = -f_\ell(t)$, it suffices to show that   $1-f_\ell(t)\ge (1-t)^{(2/3)^\ell}$ for all $t \in [0,1]$.

To prove that $1-f(t) \ge (1-t)^{2/3}$ we will use that $\arccos(t) \ge\frac{\pi(1-t)^{1/2}}{2(1+t)^{1/6}}$ for all $t \in [0,1]$ as shown in \cite[Remark 2.1]{MR3012174}. Now we have $1-f(t)  = \frac{2}{\pi}\arccos(t)\ge \frac{(1-t)^{1/2}}{(1+t)^{1/6}}$. The result now follows since $\frac{1}{(1+t)^{1/6}} \ge (1-t)^{1/6}$ for all $t \in (-1,1]$. This is the base case. The induction step uses the fact that for all $t \in [0,1]$, $f_\ell(t) \in [0,1]$. Using this and the base case proven above,
\[
1-f_{\ell+1}(t) \ge \bigl(1-f_\ell(t)\bigr)^{2/3} \ge \bigl((1-t)^{(2/3)^\ell}\bigr)^{2/3} = (1-t)^{(2/3)^{\ell+1}}.\qedhere
\]
\end{proof}

\bibliographystyle{abbrv}
\bibliography{bib}

\end{document}